\newif\ifarxiv 
 \arxivtrue

\ifarxiv
\documentclass{article}
\usepackage[left=3cm, right=3cm, bottom=4cm]{geometry}
\newif\ifCLASSOPTIONcaptionsoff
\else
\documentclass[journal]{IEEEtran}
\fi

\usepackage{epsfig} 
\usepackage{amssymb}  
\usepackage{amsmath}

\usepackage{enumitem}
\usepackage{pifont}
\usepackage[dvipsnames]{xcolor}
\usepackage{tikz}
\usepackage{pgfplots}
\usepackage{algorithm2e}
\usepackage{varwidth}
\usepackage{soul}
\usepackage{verbatim}
\usetikzlibrary{arrows,shapes}

\usepackage{graphics} 
\usepackage{amsmath} 
\usepackage{amssymb}  

\usepackage{amsthm}
\usepackage[scientific-notation=true]{siunitx}
\usepackage[colorlinks,citecolor=blue]{hyperref}
\usepackage{amsfonts}
\usepackage{tikz}
\usetikzlibrary{arrows.meta}
\usepackage{dsfont}

\usepackage{amsmath} 
\usepackage{graphicx}
\usepackage{caption}
\usepackage{subcaption}
\usepackage{multirow}

\newtheorem{theorem}{Theorem}

\newtheorem{proposition}{Proposition}
\theoremstyle{definition}
\newtheorem{example}{Example}
\newtheorem{definition}{Definition}
\newtheorem{assumption}{Assumption}
\newcommand{\PP}{\mathcal{P}}

\theoremstyle{remark}
\newtheorem{remark}{Remark}

\usepackage{placeins}

\begin{document}
%
\title{How Will the Presence of Autonomous Vehicles Affect the Equilibrium State of Traffic Networks?}
%
%
%

\ifarxiv
\author{Negar Mehr and Roberto Horowitz\footnote{N. Mehr  and R. Horowitz are with the Department
		of Mechanical Engineering, University of California, Berkeley, Berkeley,
		CA, 94720 USA e-mails: negar.mehr@berkeley.edu, horowitz@berkeley.edu }}
\fi
%
%



\maketitle

\begin{abstract}
It is known that connected and autonomous vehicles are capable of maintaining shorter headways and distances when they form platoons of vehicles. Thus, such technologies can result in increases in the capacities of traffic networks. Consequently, it is envisioned that their deployment will boost the network mobility. In this paper, we verify the validity of this impact under selfish routing behavior of drivers in traffic networks with mixed autonomy, i.e. traffic networks with both regular and autonomous vehicles. We consider a nonatomic routing game on a network with inelastic (fixed) demands for the set of network O/D pairs, and study how replacing a fraction of regular vehicles by autonomous vehicles will affect the mobility of the network. Using the well known US bureau of public roads (BPR) traffic delay models, we show that the resulting Wardrop equilibrium is not necessarily unique even in its weak sense for networks with mixed autonomy.
We state the conditions under which the total network delay is guaranteed not to increase as a result of autonomy increase. However, we show that when these conditions do not hold, counter intuitive behaviors may occur: the total delay can grow by increasing the network autonomy. In particular, 
we prove that for networks with a single O/D pair, if the road degrees of asymmetry are homogeneous, the total delay is 1) unique, and 2) a nonincreasing continuous function of network autonomy fraction. We show that for heterogeneous degrees of asymmetry, the total delay is not unique, and it can further grow with autonomy increase. We demonstrate  that similar behaviors may be observed in networks with multiple O/D pairs. We further bound such performance degradations due to the introduction of autonomy in homogeneous networks. 
\ifarxiv

\fi
\end{abstract}

%

\section{Introduction}\label{intro}
Connected and autonomous vehicles technology have attracted significant attention as a result of their potentials for increasing vehicular safety and drivers' comfort. Connected technologies can be used to inform drivers about the existing hazards through vehicle to vehicle (V2V) or vehicle to infrastructure (V2I) communication. Aligned with these safety considerations, automobile companies have started to equip vehicles with autonomous capabilities. In fact, some of these capabilities, such as driver assistive technologies and adaptive cruise control (ACC) have already been deployed in vehicles. 

The impact of these technologies is not limited to vehicles safety. Connected and autonomous vehicles technology can facilitate vehicle \emph{platooning}. Vehicle platoons are groups of more than one vehicle, capable of maintaining shorter headways; thus, platooning can lead to increases in the capacities of network links~\cite{lioris2017platoons}. Such increases can be up to three--fold~\cite{lioris2017platoons} if all the vehicles are autonomous and connected. 
In addition to mobility benefits, platooning can have sustainability benefits, it can also reduce energy consumption for heavy duty vehicles~\cite{al2010experimental,liang2013fuel,alam2015heavy}.

The mobility benefits of platooning and autonomous capabilities of vehicles are not limited to increasing network capacities. There has been a focus on how to utilize vehicle autonomy and connectedness to remove signal lights from intersections and 
coordinate conflicting movements such that the network throughput is improved~\cite{zhang2016optimal,tallapragada2015coordinated,miculescu2014polling,fayazi2018mixed}. However, in order for such approaches to be implemented, all vehicles in the network need to have autonomous capabilities. To reach the point where all vehicles are autonomous, transportation networks need to face a \emph{transient} era, when both regular and autonomous vehicles coexist in the networks. Therefore, it is crucial to study networks with mixed autonomy.

In~\cite{askari2016measuring}, the performance of traffic networks with mixed autonomy was studied via simulations. Moreover, it was shown in multiple works that in networks with mixed autonomy, autonomous vehicles can be utilized to stabilize the low--level dynamics of traffic networks and damp congestion shockwaves~\cite{wu2017emergent, darbha1999intelligent,yi2006macroscopic,pueboobpaphan2010driver,stern2018dissipation}. In~\cite{mehr2018game1,mehr2018game2} altruistic lane choice of autonomous vehicles was studied.  In~\cite{lazar2017capacity}, the capacity of network links was modeled in a traffic setting with mixed autonomy. This modeling framework was further used in~\cite{lazar2017price} to calculate the price of anarchy of traffic networks with mixed autonomy, where the price of anarchy is an indicator of how far the equilibrium of networks with mixed autonomy is from their social optimum that could have been achieved if a social planner had routed all the vehicles. In~\cite{lazar2018maximizing}, it was shown that local actions of the autonomous vehicles on the road
can lead to optimal vehicle orderings for the global network properties such as link capacities.

It is well known that due to the selfish route choice behavior of drivers, traffic networks normally operate in an equilibrium state, where no vehicle can decrease its trip time by unilaterally changing its route~\cite{smith1979existence}. 
In this paper, we wish to study how the introduction of autonomous vehicles in the network will affect the equilibrium state of traffic networks compared to the case when all vehicles are nonautonomous. We extend our initial results presented in~\cite{mehr2018can}. In particular, given a fixed demand of vehicles, we study how replacing a fraction of regular vehicles by autonomous vehicles will affect the equilibrium state of traffic networks. We study the system behavior when both regular and autonomous vehicles select their routes \emph{selfishly} to investigate the necessity of centrally enforcing autonomous vehicles routing by a network manager. We state the conditions under which increasing network autonomy fraction is guaranteed to reduce the overall network delay. Moreover, we show that when these conditions do not hold, counter intuitive and undesirable behaviors might occur, such as the case when increasing the portion of autonomous vehicles in the network can \emph{increase} the overall network delay. Such behaviors are similar to Braess' paradox, where the construction of a new road or expanding link capacities may increase total network delay. 

We model the network in a macroscopic framework where vehicle route choices are taken into account. We model the selfish route choice behavior of the drivers as a nonatomic routing game~\cite{roughgarden2002selfish} where drivers choose their routes selfishly until a Wardrop Equilibrium is achieved~\cite{wardrop1952some}. We represent a traffic network by a directed graph with a certain set of origin destination (O/D) pairs. For each O/D pair, we consider two classes of vehicles, regular and autonomous. For a given fixed demand profile along O/D pairs, we study how increasing the autonomy fraction of O/D pairs will affect the total delay of the network at equilibrium. 

We first show that the equilibrium may \emph{not} be unique even in the weak sense of total link utilization. 
Then, we study networks with a single O/D pair and prove that if the degrees of road capacity asymmetry are homogeneous in the network, the social or total delay of the network is unique, and further it is a monotone nonincreasing function of the network autonomy ratio. However, in networks with heterogeneous degrees of road asymmetry, we first
show that the social delay is not unique. Then, we demonstrate that, surprisingly, increasing the autonomy ratio of the network may lead to an \emph{increase} in the overall network delay. This is a counter intuitive behavior as we might expect that having more autonomous vehicles in the network will always be beneficial in terms of total network delay. For the networks with multiple O/D pairs, we show that similar complicated behaviors may occur, namely increasing autonomy fraction of an O/D pair might worsen the social delay of the network. Our work in fact shows that traffic paradoxes similar to the well known Braess's Paradox~\cite{braess1968paradoxon} can occur due to capacity increases provided by autonomous vehicles. We further bound such performance degradations that can arise from the presence of autonomous vehicles.

The organization of this paper is as follows. In Section~\ref{sec:model}, we describe our notation and model. We review the prior relevant results in Section~\ref{sec:uniq}. Then, in Section~\ref{sec:uniq_mixed}, we study the uniqueness of equilibrium in our routing setting. Next, in Section~\ref{sec:single_OD}, we analyze mixed autonomy networks with a single O/D pair in Section~\ref{sec:single_OD}. Subsequently, we study mixed--autonomy networks with multiple O/D pairs in Section~\ref{sec:mult}. Finally, we conclude the paper and provide relevant future directions in Section~\ref{sec:future}.







\section{Nonatomic Selfish Routing} \label{sec:model}

We model a traffic network by a directed graph $G = (N,L,W)$, where $N$ and $L$
are respectively the set of nodes and links in the network.
Each link $l \in L$ in the network is a pair of distinct nodes $(v,w)$ and represents a directed edge from $v$ towards $w$.
We assume that each link joins two distinct nodes; thus, no self loops are allowed.
Define $W = \{(o_1,d_1), (o_2,d_2),\cdots, (o_k,d_k)\}$ to be the set of origin destination (O/D) vertex pairs of the network. 
 A node $n \in N$ can appear in multiple O/D pairs. In a nonatomic selfish routing game,
if each O/D pair has a fixed given
  nonzero demand, then it is called a nonatomic selfish routing game with inelastic demands.
Each O/D pair consists of infinitesimally small agents where every agent decides on each path such that their own delay is minimized.
The delay of each
path depends on how network paths are shared among different O/D pairs. For each O/D pair $w = (o_i,d_i), 1\leq i \leq k$, we let $\PP_w$ denote the set of all possible network paths from $o_i$ to $d_i$. We assume that the network topology is such that for each O/D pair $w \in W$, there exists at least one path from its origin to its destination, i.e. $\PP_w \neq \emptyset$. We further let $\PP = \cup_{w\in W} \PP_w$ denote the set of all network paths.

For an O/D pair $w \in W$, let $r_w$ be the given fixed demand of vehicles
associated with $w$. Furthermore, for a path $p \in \PP_w$, let $f_p$ be the
flow of the O/D pair $w$ along path $p$. Note that each path connects exactly one origin to one and only one destination; thereby, once a path is fixed, its origin and destination are uniquely determined. Consequently, there is no need to explicitly include path O/D pairs in the notation used for $f_p$. It is important to note that in our setting, 
each O/D pair $w$ has two classes of vehicles: autonomous and regular.
Consequently, for each $w \in W$, we define $\alpha_w$ to be the fraction of vehicles in $r_w$ that are autonomous. We let $r = {(r_w: {w \in W})}$ and $\alpha = (\alpha_w:{w \in W})$ be the vectors of network demand and autonomy fraction respectively. Also, for each path $p \in \mathcal{P}_w$, we use $f^r_p$ and $f^a_p$ to respectively denote the flow of regular and autonomous vehicles along path $p$. Note that for each path $p \in \PP$, we have $f_p = f_p^r + f_p^a$. Moreover, for each O/D pair $w \in W$, due to flow conservation, we must have $\sum_{p \in \PP_w} f^r_p = r_w (1-\alpha_w) $, and $\sum_{p \in \PP_w  } f^a_p = r_w \alpha_w$. The network flow vector $f$ is a nonnegative vector of regular and autonomous flows along network paths, i.e. $f = (f_p^r,f_p^a: {p\in \PP})$. A  flow vector $f$ is called feasible for a given network $G$, if for each $w \in W$,

\begin{subequations}\label{eq:feasibility}
\begin{equation}
  \sum_{p \in \PP_w} f_p^r= (1-\alpha_w)r_w,\; \text{and} \; \sum_{p \in \PP_w} f_p^a= \alpha_w r_w,
\end{equation}    
\begin{equation}
   f^r_p \geq 0,\; \text{and}\; f^a_p \geq 0,\; \forall p \in \PP_w.
\end{equation}
\end{subequations}


For each link $l \in L$, $f_l$ is the total flow of vehicles in link $l$, i.e.\ $f_l = \sum_{p \in \PP: l \in p}f_p$. Since we need to decompose the total link flow into regular and autonomous vehicles, we let
$f_l^r$ and $f_l^a$ be the total flow of regular and autonmous vehicles along link $l$ respectively. In fact, $f_l^r$ and  $f_l^a$ are the summation of the flow of regular and autonomous vehicles on all routes containing link $l$,
\begin{align*}
f^r_l = \sum_{p \in \mathcal{P}: l\in p} f^r_p,\,\,\,\text{and} \,\,\, f^a_l = \sum_{p \in \mathcal{P}: l\in p} f^a_p.
\end{align*} 

Note that if all vehicles are regular for an O/D pair $w \in W$, i.e. $\alpha_w = 0$, then, we only have a single class of regular vehicles along that O/D pair, and for each path $p \in \PP_w$, $f_p = f_p^r$. 
If for all network O/D pairs $w \in W$, the autonomy fraction $\alpha_w = 0$; then, the same argument holds for link flows, $f_l = f_l^r$ for all links $l \in L$. In fact, if all vehicles are regular, our routing game reduces to a single class game
\begin{align}\label{eq:single_type}
\left( \forall w \in W,\; \alpha_w = 0 \right) \Longleftrightarrow \left(\forall p \in \PP,\; f_p = f_p^r \right).
\end{align}

In order to be able to model the incurred delays when vehicles are routed throughout the network, it is assumed that each link $l \in L$ has a delay per unit of flow function $e_l: \mathbb{R}^2 \rightarrow \mathbb{R}$. We assume that the delay per unit of flow for each path $p \in \PP$ is obtained by the summation of the link delays over the links that form $p$,
\begin{align}
e_p(f) = \sum_{l \in L: l \in p} e_l(f^r_l,f^a_l). \label{eq:route_cost}
\end{align}
Equation~\eqref{eq:route_cost} implies that the delay of each path $p \in \PP$ depends not only on the flows of regular and autonomous vehicles along path $p$, but also on the flows along other paths. The overall network delay or social delay is given by
\begin{align}
J(f) = \sum_{p\in \PP} f_p e_p(f).
\end{align}


\subsection{Wardrop Equilibrium}
It is well known in the transportation literature that if there are many noncooperative agents, namely, flows that behave selfishly~\cite{roughgarden2006severity}, a network is at an equilibrium if the well known Wardrop conditions hold~\cite{wardrop1952some}. The Wardrop conditions state that at equilibrium, no user has any incentive for unilaterally changing its path. This implies that for an equilibrium flow vector $f$, if there exists a path $p \in \mathcal{P}_w$ such that either $f_p^r \neq 0$ or $f_p^a \neq 0$, we must have that $e_p(f) \leq e_{p'}(f)$, for all paths $p' \in \PP_w$.

\begin{definition}
Given a network $G = (N,L,W)$, a flow vector $f$ is a Wardrop equilibrium if and only if for every O/D pair $w \in W$ and every $p,p' \in \mathcal{P}_w$
\begin{subequations}
\label{eq:equilib_conds}
\begin{gather}
\begin{align}
f^r_p \left(e_p(f)-e_{p'}(f)\right) &\leq 0, \\
f^a_p \left(e_p(f)-e_{p'}(f) \right) &\leq 0. 
\end{align}
\end{gather}
\end{subequations}
\end{definition}


\noindent Note that an implication of the above definition is that for each O/D pair $w \in W$, and any two paths $p,p' \in \PP_w$ such that $f_p \neq 0$ and $f_{p'}\neq 0$, we must have that $e_p(f) = e_{p'}(f)$. 

\begin{definition}
Given an equilibrium flow vector $f$ for the network $G = (N, L, W)$, we define the delay of travel for each O/D pair $w \in W$ to be 
\begin{equation}\label{eq:delay_travel}
e_w(f) := \min_{p \in \PP_w} e_p(f).
\end{equation}
\end{definition}

\noindent Motivated by the above discussion, $e_w(f)$ is precisely the delay across all paths $p \in \PP_w$ which have a nonzero flow. Moreover, the equilibrium condition implies that for a path $p \in \PP_w$ with zero flow, we have $e_p(f)  \geq e_w(f)$. 

It is worth mentioning that when there are no autonomous vehicles, i.e. for all $ w \in W, \alpha_w = 0$, since $f_p^r = f_p$ for all $p \in \PP$, Conditions~\eqref{eq:equilib_conds} reduce to
\begin{align}\label{eq:equilib_cond_single_com}
f_p \left(e_p(f)-e_{p'}(f)\right)\leq 0,\quad \forall w \in W, \;\forall p,p' \in \PP_w.
\end{align}

\subsection{Delay Characterization}

We first specify the structure of our delay functions. If there is only a single class of regular vehicles in the network, the US Bureau of Public Roads (BPR)~\cite{manual1964bureau} suggests the following form of delay functions.

\begin{assumption}
When network links are shared by only regular vehicles, the link delay functions $e_l: \mathbb{R} \rightarrow \mathbb{R}$ are of the following form

\begin{align}\label{eq:latency_func}
e_l(f_l) = a_l \left(1 + \gamma_l \left(\frac{f_l}{C_l}\right)^{\beta_l}\right),
\end{align}

\noindent where $C_l$ is the capacity of link $l$, and $a_l, \gamma_l,$ and $\beta_l$ are nonnegative link parameters. 
\end{assumption}

\noindent In practice, $a_l$ is the free flow travel time on $l$, $\gamma_l$ is normally $0.15$, and $\beta_l$ is a positive integer ranging from 1 to 4. 
In order to characterize the delay functions in networks with mixed autonomy, where we have two classes of vehicles, we first need to model the impact of autonomous vehicles on link capacities. In each network link $l \in L$, the link capacity $C_l$ restricts the maximum possible flow of vehicles. It was shown in~\cite{lazar2017capacity} that in networks with mixed autonomy, $C_l$ depends on the autonomy ratio of link $l$ defined as $\alpha_l := \frac{f_l^a}{f_l^a+f_l^r}$. We use $C_l(\alpha_l)$ to emphasize this dependence. Let $m_l$ and $M_l$ be the capacity of link $l$ when all vehicles are regular and autonomous respectively. Since autonomous vehicles are capable of maintaining shorter headways, it is normally the case that $ \frac{m_l}{M_l}\leq 1$. When the two classes of regular and autonomous vehicles are present in the network, using the results in~\cite{lazar2017capacity}, we have 

\begin{align}\label{eq:capac_model}
C_l(\alpha_l) &= \frac{m_l M_l}{\alpha_l m_l + (1-\alpha_l)M_l}. 
\end{align}

We adopt this model throughout this paper to investigate the mobility impact of autonomous vehicles on the network. Since for each link $l \in L$, $\alpha_l = \frac{f_l^a} {f_l^a+f_l^r}$ and $f_l = f_l^a + f_l^r$, using~\eqref{eq:capac_model}, for networks with mixed autonomy, the delay function~\eqref{eq:latency_func} can be modified as:

\begin{align}
e_l(f^r_l, f^a_l) &= a_l \left(1 + \gamma_l \left(\frac{f^r_l + f^a_l}{\frac{m_l M_l (f_l^r + f_l^a)}{m_l f_l^a + M_l f_l^a}} \right)^{\beta_l}\right).
 \\
&=
a_l \left(1 + \gamma_l \left(\frac{f^a_l}{M_l}+\frac{f_l^r}{m_l}\right)^{\beta_l}\right).\label{eq:latency_simplified}
\end{align}

Note that when only regular vehicles are present in the network, for each link $l \in L$ since $f_l = f_l^r$, the link delay function reverts to

\begin{align}\label{eq:latency_simplified_no_aut}
e_l(f_l) = a_l \left(1 + \gamma_l \left(\frac{f_l^r}{m_l}\right)^{\beta_l}\right).
\end{align}

 


\section{Prior Work }\label{sec:uniq}
\begin{figure}
\centering
\begin{tikzpicture}[scale=0.8]
\begin{scope}[every node/.style={circle,thick,draw}]
    \node (A) at (0,0) {A};
    \node (B) at (3,1.5) {B};
    \node (C) at (3,-1.5) {C};
    \node (D) at (6,0) {D};
    
\end{scope}

\begin{scope}[>={Stealth[teal]},
              every node/.style={fill=white,circle},
              every edge/.style={draw=teal,very thick}]
    \path [->] (A) edge node {$1$} (B);
    \path [->] (A) edge node {$3$} (C);
    \path [->] (B) edge node {$2$} (D);
    \path [->] (C) edge node {$4$} (D);
\end{scope}

\end{tikzpicture}
\caption{A network with a single O/D pair and two paths.}\label{fig:net_equ_1}
\end{figure}
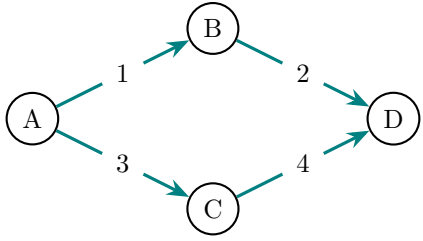

\subsection{Existence of Equilibrium}
We state the following proposition from~\cite{braess1979existence}
which studies the conditions under which a Wardrop Equilibrium exists for a multiclass traffic network.

\begin{proposition}\label{prop:exist}
Given a network $G = (N,L,W)$, if the link delay functions are continuous and monotone in the link flow of each class; then, there exists at least one Wardrop equilibrium. 
\end{proposition}

\begin{remark}
\label{remark:exist}
Using~\eqref{eq:latency_simplified}, since our assumed delay functions are nonnegative, continuous, and monotone in the flow of each class, Proposition~\ref{prop:exist} implies that there always exists at least one Wardrop equilibrium for a routing game with mixed autonomy. 
\end{remark}

\subsection{Uniqueness of Equilibrium}
In this part, we review the known results regarding the uniqueness of the Wardrop Equilibrium. When multiple classes of vehicles are present in the network, the uniqueness of the equilibrium flow vector does not hold. However, uniqueness in a weak sense is known to hold from~\cite{altman2006survey}.


%
 
\begin{proposition}\label{prop:weak_uniq}For a general topology network $G$ with multiple classes of vehicles on each O/D pair, if the delay functions are of the form~\eqref{eq:latency_func}, and the link capacities $C_l$ are fixed and the same for all vehicle classes, for a given demand vector $r$, we have

\begin{enumerate}
\item The equilibrium is unique in a weak sense, i.e. for each link $l$, the
  total flow $f_l$ for all Wardrop equilibrium flow vectors $f$ is unique. 
\item For each O/D pair $w \in W$, the delay of travel $e_w(f)$ is unique for
  all Wardrop equilibrium flow vectors $f$. Thus, the delay of travel for each
  O/D pair in equilibrium, i.e. $e_w(f)$,  only depends on  the network demand vector $r$.
 Hence, we may unambiguously define  $e_w(r)$ to denote this unique value.
\end{enumerate} 
\end{proposition}

\begin{remark}\label{rem:sing_type}
Note that a routing game that has only a single class of vehicles can be viewed as an instance of the games described in Proposition~\ref{prop:weak_uniq}. Therefore, uniqueness in a the weak sense applies to games with a single class of vehicles too. 
\end{remark}

\subsection{Monotonicity of Social Delay}

As we discussed above, in general, the equilibrium is not unique. However, if the conditions of Proposition~\ref{prop:weak_uniq} hold for a network, the social delay and the delay of travel for each O/D pair  are unique.
For a single class routing game on $G = (N, L, W)$, 
recall the following from~\cite{hall1978properties}.  
\begin{proposition}\label{prop:contin_latency}
Consider a network $G = (N,L,W)$, where only one class of vehicles exists for each O/D pair $w \in W$. Assume that for  each link $l \in L$, $e_l(.)$ is continuous, positive valued, and monotonically increasing. Then, for each $w \in W$, the delay of travel $e_w(r)$ is a continuous function of the demand vector $r$. Furthermore, $e_w(.)$ is nonincreasing in $r_w$ when all other demands $r_{w'}, w' \neq w$, are fixed. 
\end{proposition}





\section{Uniqueness in the Mixed-Autonomy Setting}\label{sec:uniq_mixed}

\indent Now we study equilibrium uniqueness in our setting. Using Remark~\ref{remark:exist}, we know that there exists at least one equilibrium. However, since in our setting, for each link $l$, $C_l$ depends on the autonomy ratio $\alpha_l$, Proposition~\ref{prop:weak_uniq} does not apply. Indeed, we demonstrate through an example that the equilibrium is \emph{not} unique even in the weak sense introduced in Proposition~\ref{prop:weak_uniq}. 

\begin{example}
 Consider the network of Figure~\ref{fig:net_equ_1}. Let $p_1$ and $p_2$ be the ABD and ACD paths respectively. For each link $l = 1,\cdots, 4$, let the link parameters be $\beta_l = 1, a_l = 1, m_l = 1$, and $, M_l = 2$. Thus, for each link $l \in L$, the link delay function is $e_l = 1+f_l^r+\dfrac{f_l^a}{2}$. Assume that the demand from node A to D is $r=2$, and $\alpha = 0.5$. The example is simple enough so that we can compute the equilibrium flows manually. Let $f_1^r$ and $f_1^a$ be the regular and autonomous vehicles flows along $p_1$, and $f_2^r$ and $f_2^a$ be the regular and autonomous flows along $p_2$. At equilibrium, using the symmetry of the network, we must have
\begin{align*}
2 + 2f_1^r + f_1^a &= 2 + 2f_2^r + f_2^a \\
f_1^r + f_2^r &= 1 \\
f_1^a + f_2^a &= 1 \\
f_1^r, f_1^a, f_2^r, f_2^a &\geq 0.
\end{align*}
Clearly, there is no unique solution to the above set of equations. Moreover, among the set of all possible equilibrium flow vectors, for each link, the maximum link flow at equilibrium is 1.25, whereas the minimum link flow is 0.75 at equilibrium. 
This implies that equilibrium uniqueness does not hold even in the weak sense for traffic networks with mixed autonomy.  

\end{example}

\section{Networks with a Single O/D Pair} \label{sec:single_OD}

\begin{figure}
\centering
\begin{tikzpicture}[scale=0.8]
\begin{scope}[every node/.style={circle,thick,draw}]
    \node (A) at (0,0) {A};
    \node (B) at (3,1.5) {B};
    \node (C) at (3,-1.5) {C};
    \node (D) at (6,0) {D};
    
\end{scope}

\begin{scope}[>={Stealth[teal]},
              every node/.style={fill=white,circle},
              every edge/.style={draw=teal,very thick}]
    \path [->] (A) edge node {$1$} (B);
    \path [->] (A) edge node {$3$} (C);
    \path [->] (B) edge node {$2$} (D);
    \path [->] (C) edge node {$4$} (D);
    \path [->] (B) edge node {$5$} (C);
\end{scope}

\end{tikzpicture}
\caption{A network with a single O/D pair and three paths from A to D.}
\label{fig:Braes}
\end{figure}
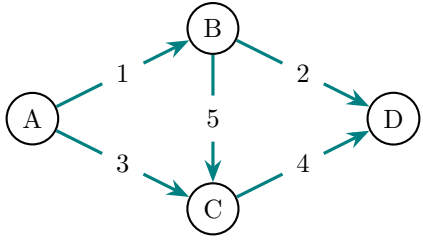

In this section, we study two terminal networks which have a single O/D pair in the presence of autonomy. For such networks, since there is only one O/D pair, all paths originate from a common source $o$ and end in a common destination $d$. Since $W$ is singleton, we omit the subscript $w$ from $r_w$, $e_w$ and $\alpha_w$ throughout this section. Note that when the network has a single O/D pair, $r$ and $\alpha$ are scalars.

Having observed that in the mixed-autonomy setting, the equilibrium is not unique even in the weak sense, it is important to study if the social delay is unique for all network equilibrium flow vectors. To this end, in the following,  
we study the properties of the social delay including its uniqueness. To this end, we need to define the notion of road degree of capacity asymmetry introduced in~\cite{lazar2017price}. Given a network $G=(N,L,W)$, for each link $l \in L$, we define $\mu_l := m_l/M_l$ to be the degree of capacity asymmetry of link $l$. Note that since we assumed that autonomous vehicles headway is less than or equal to that of regular vehicles, for each link $l \in L$, $\mu_l \leq 1$. In the sequel, we consider two scenarios for investigating the properties of social delay: 

\begin{enumerate}
\item Homogeneous degrees of road capacity asymmetry, where $\mu_l$ is the same for all links, i.e. $\mu_l = \mu$, for all links $l \in L$, where $\mu$ is the common value of capacity asymmetry.
\item Heterogeneous degrees of capacity asymmetry, where $\mu_l$ varies on different links. 
\end{enumerate}

\subsection{Homogeneous Degrees of Capacity Asymmetry}

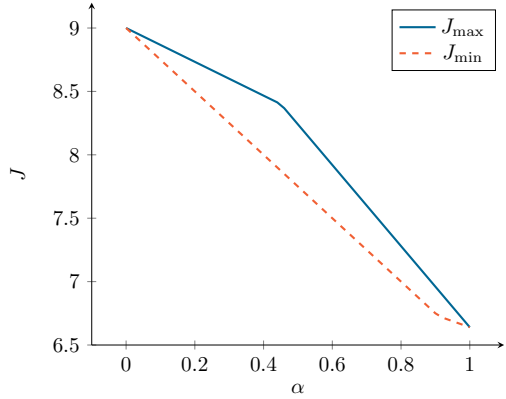
\begin{figure}
\centering
 \begin{tikzpicture}[scale=0.8]
\begin{axis}
[
no markers,
axis x line=bottom,
axis y line=left,
xmax=1.1, ymax = 9.2,
xmin=-0.1, ymin=6.5,
legend entries = {$J_\text{max}$, $J_\text{min}$},
xlabel={$\alpha$},
ylabel={$J$},
]
\addplot[MidnightBlue,line width=1pt] table [x=a, y=max, col sep=comma,mark = none] {f0_2__a0_1_2_3_1_1__m0_1_1_1_1_1__M0_1_3_3_2_4___50-points.txt};
\addplot[RedOrange, dashed,line width=1pt] table [x=a, y=min, col sep=comma, mark = none] {f0_2__a0_1_2_3_1_1__m0_1_1_1_1_1__M0_1_3_3_2_4___50-points.txt};
\end{axis}
\end{tikzpicture}
\caption{Maximum and minimum social delay for Example~\ref{ex:non_unique_delay}.}\label{fig:dec_Js}
\end{figure}

In this case, we can establish the uniqueness of the social delay, and characterize the relationship between social delay and network autonomy ratio. 

\begin{theorem}\label{theo:homog}
Given a network $G = (N, L, W)$ with a single O/D pair and a
homogeneous degree of capacity asymmetry $\mu$, for any demand vector $r > 0$, we have:
\begin{enumerate}
\item For a fixed autonomy ratio $0 \leq \alpha\leq 1$, the social delay $J(f)$ is unique for all Wardrop equilibrium flow vectors $f$. 
\item If for each $0 \leq \alpha \leq 1$, we denote the common value of social delay in the above by $J(\alpha)$, then $J(.)$ is continuous and nonincreasing.
\end{enumerate}
\end{theorem}
\begin{proof}
Fix $r >0$ and $0 \leq \alpha \leq 1$. Recalling Remark \ref{remark:exist}, we know that a Wardrop equilibrium exists. Let $f = (f^r_p, f^a_p: p \in \PP)$ be such an equilibrium flow vector where $f_p = f_p^a + f_p^r$ for each path $p$ in $\PP$. Define  $e_\text{min}(f) := \min_{p \in \PP} e_p(f)$.
Since the network has only one O/D pair, and the delay associated with all paths
with nonzero flows are the same, denoting this uniform path delay by
$e_\text{min}(f)$, we realize that the social delay is given by $J(f) = r
e_\text{min}(f)$. For each path $p \in \PP$, define the fictitious single-class
regular flow $\tilde{f}_p := f^r_p + \mu f_p^a$. We claim that the flow vector
$\tilde{f} = (\tilde{f}_p: p \in \PP)$ is  a Wardrop equilibrium for a routing
game on $G$ with a single class of regular vehicles and a total demand of
$\tilde{r} = r(1-\alpha) + r \alpha \mu$ with the delay function
  $(\tilde{e}_l: l \in L)$ defined as
  \begin{equation*}
    \tilde{e}_l(\tilde{f}_l) = a_l \left (1+\gamma_l\left(\frac{\tilde{f}_l}{  m_l}\right) ^{\beta_l}\right).
  \end{equation*}


To see this, for each $p\in \PP$, we show that relations \eqref{eq:equilib_cond_single_com} hold. Fix $p, p' \in \PP$ and note that since $f$ was a Wardrop equilibrium in the original setting, we have $f^r_p (e_p(f) - e_{p'}(f)) \leq 0$, and $f^a_p (e_p(f) - e_{p'}(f)) \leq 0$. 
Multiplying the latter by the positive constant $\mu$ and adding the two inequalities, we have 
\begin{equation}
\label{eq:tilde-f-e}
\tilde{f}_p (e_p(f) - e_{p'}(f)) \leq 0, \qquad \forall p,p' \in \PP.
\end{equation}
Now, we claim that for all $p \in \PP$, we have $e_p(f) = \tilde{e}_p(\tilde{f})$. Note that for each link $l \in L$, we have $\tilde{f}_l = f_l^r + \mu f_l^a$. Using the fact that $\mu = m_l / M_l$ for all $l \in L$, we get
\begin{equation}
\label{eq:tilde-e-e}
\begin{aligned}
\tilde{e}_p(\tilde{f}) &= \sum_{l \in p} a_l \left ( 1 + \gamma_l \left ( \frac{f_l^r + \frac{m_l}{M_l} f_l^a}{m_l} \right )^{\beta_l} \right )  \\
&= \sum_{l \in p} a_l \left ( 1 + \gamma_l \left ( \frac{f_l^r}{m_l} + \frac{f_l^a}{M_l} \right )^{\beta_l} \right ) = e_p(f).
\end{aligned}
\end{equation}

\noindent Substituting into \eqref{eq:tilde-f-e}, we realize that

\begin{equation}
\tilde{f}_p (\tilde{e}_p(\tilde{f}) - \tilde{e}_{p'}(\tilde{f})) \leq 0, \qquad \forall p,p' \in \PP,
\end{equation}
which means that $\tilde{f}$ is an equilibrium flow vector. 
 Clearly, the total demand of this new routing game is $\tilde{r} = \sum_{p \in \PP} \tilde{f}_p = \sum_{p \in \PP} f^r_p + \mu f^a_p = (1-\alpha)r+\mu \alpha r$.
Moreover, define $\tilde{e}_\text{min}(\tilde{f})$ to be  the minimum of $\tilde{e}_p(\tilde{f})$ among $p\in \PP$. Since $w$ is the single O/D pair of the network, $\tilde{e}_\text{min}(\tilde{f})$ is indeed equal to $\tilde{e}_w(\tilde{f})$, the travel delay of the single O/D pair of the network associated with $\tilde{f}$. Note that Proposition~\ref{prop:weak_uniq} implies that $\tilde{e}_\text{min}(\tilde{f})$ is a function of $\tilde{r}$ only. On the other hand, \eqref{eq:tilde-e-e} implies that $\tilde{e}_\text{min}(\tilde{f}) = e_\text{min}(f)$. Putting these together, we realize that
\begin{equation*}
J(f) = r e_\text{min}(f) = r \tilde{e}_\text{min}(\tilde{f}) = r \tilde{e}_w(\tilde{r}).
\end{equation*}
Note that the right hand side of the above identity does not depend on $f$, which establishes the proof of the first part. In fact, this shows that 
\begin{equation*}
J(\alpha) = r \tilde{e}_w(r (1-\alpha) + \alpha \mu r).
\end{equation*}
From Proposition~\ref{prop:contin_latency}, we know that $\tilde{e}_w(.)$ is continuous and nonincreasing. Also, since $\mu \leq 1$, the map $r \mapsto r (1-\alpha) + \alpha \mu r$ is continuous and nonincreasing. This completest the proof of the second part. 
\end{proof}


\subsection{Heterogeneous Degrees of Capacity Asymmetry}

Now, we allow $\mu_l$ to vary among the network links. We show that this makes the behavior of the system more complex. First, we show via the following example that the social delay is not necessarily unique in this case. 
\begin{example}\label{ex:non_unique_delay}
Consider the network shown in Figure~\ref{fig:Braes}. Assume that $\gamma_l = 1, \beta_l = 1$, for $l= 1,2,\cdots, 5$. Let the other link parameters be the following: $\{a_1 = 1, m_1 = 1, M_1 = 1\}$, $\{a_2 = 2, m_2 = 1, M_2 = 3 \}$, $\{a_3 = 1, m_3 = 1, M_3 = 2\}$, $\{a_4 = 1, m_4 = 1, M_4 = 4\}$, and $\{a_5 = 3, m_5 = 1, M_5 = 3\}$. Moreover, let the total flow from origin A to destination D be 2. Now, if we compute the social delay for this network for any $\alpha > 0$ at the different equilibria of the system, we observe that the social delay is \emph{not} unique. In particular, Figure~\ref{fig:dec_Js} shows the plots of the maximum and minimum social delay of the system at equilibrium for every value of $\alpha$. As Figure~\ref{fig:dec_Js} shows, as soon as $\alpha$ starts to increase from 0, uniqueness of social delay is lost. Once, $\alpha = 1$, the uniqueness of social delay is again preserved. This behavior implies that the change in the social delay due to increasing the autonomy ratio of the network is dependent on \emph{which} equilibrium the system will be at.
\end{example}

\begin{figure}
\centering
 \begin{tikzpicture}[scale=0.9]
\begin{axis}
[
no markers,
axis x line=bottom,
axis y line=left,
xmax=1.1, ymax = 553,
xmin=-0.1, ymin=542,
legend entries = {$J_\text{max}$, $J_\text{min}$},
legend style={at={(0.7,0.93)},anchor=west},
xlabel={$\alpha$},
ylabel={$J$},
]


\addplot[MidnightBlue,line width=1pt] coordinates { (0,2724/5) (1/3,552) (1,39060/71) };
\addplot[RedOrange,dashed,line width=1pt] coordinates { (0,2724/5) (88/123,22260/41) (1,39060/71) };

\end{axis}
\end{tikzpicture}
\caption{Maximum and minimum social delays for the Example~\ref{ex:counter_intuitive}.}
\label{fig:counter-intuitive}
\end{figure}
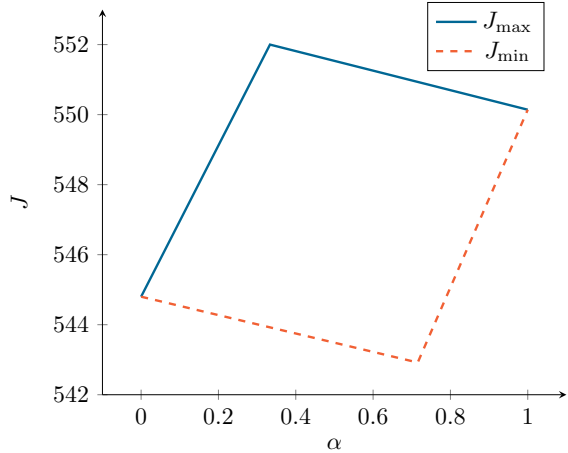

Now, we study the effect of increasing network autonomy on the social delay. In the previous example, both the maximum and minimum social delays decreased as a function of $\alpha$. But, is this necessarily the case? We use the following examples to demonstrate that it may not be true in general, as increasing network autonomy may increase social delay in some networks. 

\begin{example}
Consider the network of Figure~\ref{fig:Braes}. Let $\gamma_l = 1$ and $\beta_l=1$ for all links. Select the other network parameters to be the following, $\{a_1 =0, m_1 = 0.1, M_1 = 0.1 \}$, $\{a_2 = 50, m_2 = 1, M_2 = 1 \}$, $\{a_3 = 50, m_3 = 1, M_3 = 1 \}$, $\{ a_4 =0, m_4 = 0.1, M_4 = 0.1 \}$, $\{a_5 = 10, m_5 = 0.5, M_5 = 1 \}$. Let the total O/D demand be  $r= 6$. In the absence of autonomy ($\alpha = 0$), the social delay is $J = 504.3$. However, if we increase the autonomy ratio to $\alpha = \frac{1}{10}$, $J = 518.6$. Clearly, in this case, the social delay increases when the network autonomy ratio $\alpha$ is increased. Note that since $\mu_l = 1$ for $ l = 1,2,3,4$ and $\mu_5 = 0.5 < 1$, this can be viewed as an instance of the classical Braess's Paradox~\cite{braess1968paradoxon}, where an increase in the capacity of the middle link of a Wheatstone network can paradoxically lead to an increase in the social delay.

\end{example}

 One might argue that if we allow $\mu_l$ to be strictly less than 1 for all network links $l \in L$, the network social delay will decrease. We use the following example to show that even in this case, increasing autonomy can worsen social delay.
 
\begin{example}\label{ex:counter_intuitive}
Consider the previous example with the total flow $r = 6$, but change $M_l$'s to be, $M_1 = \frac{1}{9},\, M_2 = 1.1,\, M_3 = 1.1,\, M_4 = \frac{1}{9},$ and $M_5 = 1$. In this case, clearly, $\mu_l < 1,$ for all $ l \in L$. We computed the maximum and minimum social delay at equilibrium for every autonomy fraction $\alpha$. Figure~\ref{fig:counter-intuitive} shows the maximum and minimum social delay in this example for different values of $\alpha$. Figure~\ref{fig:counter-intuitive} demonstrates that the maximum social delay increases as we increase $\alpha$ from 0, until we reach a local maximum. The minimum social delay decreases as we increase $\alpha$ from 0, until we reach a local minimum, and then, it increases sharply to values that are higher than the social delay at $
\alpha = 0$. Surprisingly, when all vehicles are autonomous ($\alpha=1$) the social delay is greater than the social delay when $\alpha=0$, i.e. $J( \alpha =
 1) > J(\alpha = 0)$. This might be counter intuitive as we expect the network with full autonomy to have smaller social delay. However, this example shows that when capacity increases are heterogeneous across the network, the selfish behavior of the vehicles when making their route choices might actually lead to worsening the social delay of the network. Therefore, the mobility benefits obtained from the introduction of autonomous vehicles in the network, in terms of decreasing network social delay, are not obvious. 

\end{example}

 As mentioned previously, the increase in social delay due to an increase in the fraction of autonomous vehicles is in fact a particular instance of Braess's paradox. Braess's Paradox is the counterintuitive but well known fact that removing edges from a network or increasing the delay functions on certain links can improve social delay~\cite{roughgarden2006severity}. In our problem setting, replacing a fraction of regular vehicles by autonomous vehicles can be interpreted as replacing the link delay function $a_l \left(1 + \gamma_l \left(\frac{f^a_l}{m_l}+\frac{f_l^r}{m_l}\right)^{\beta_l}\right)$ by $a_l \left(1 + \gamma_l \left(\frac{f^a_l}{M_l}+\frac{f_l^r}{m_l}\right)^{\beta_l}\right)$ for every link $l \in L$.
 It was shown in previous studies that Braess paradox is prevalent and can be arbitrarily severe~\cite{steinberg1983prevalence, roughgarden2006severity}. Despite the price of anarchy, the occurence of Braess's paradox heavily depends on network topology and the parameters of link delay functions ~\cite{roughgarden2001designing,hagstrom2001characterizing,milchtaich2003network}. 

\section{Networks with Multiple O/D Pairs}\label{sec:mult}

So far, we have seen that even in a network with only one O/D pair, the introduction of autonomous vehicles can result in complex behaviors. Thus, it should be expected that a general network with multiple O/D pairs will exhibit similar counter intuitive behaviors. In the previous section, we saw that the existence of a homogeneous degree of capacity asymmetry throughout the network is sufficient for guaranteeing improvements in the social delay by increasing the fraction of autonomous vehicles. We now show, via the following example, that this is not the case for networks with multiple O/D pairs.

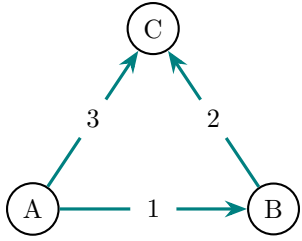
\begin{figure}
\centering
\begin{tikzpicture}[scale=0.8]
\begin{scope}[every node/.style={circle,thick,draw}]
    \node (A) at (0,0) {A};
    \node (B) at (4,0) {B};
    \node (C) at (2,3) {C};
    
\end{scope}

\begin{scope}[>={Stealth[teal]},
              every node/.style={fill=white,circle},
              every edge/.style={draw=teal,very thick}]
    \path [->] (A) edge node {$1$} (B);
    \path [->] (A) edge node {$3$} (C);
    \path [->] (B) edge node {$2$} (C);
\end{scope}

\end{tikzpicture}
\caption{A network with three O/D pairs.} 
\label{fig:net_equ}
\end{figure}

\begin{example}
Consider the network shown in Figure~\ref{fig:net_equ} which was first introduced in~\cite{fisk1979more}. There are three O/D pairs, $W = \{ \text{(A,B), (B,C), (A,C)}\}$. The total demand of the network O/D pairs are $r_{\text{AB}} = 17, r_{\text{AC}} = 20$, and $r_{\text{BC}} = 90 $. Assume that $\gamma_l=1,\,\beta_l=1$, for all links $l \in L$. Let the link parameters be $\{a_1 = 0, m_1 = 1, , M_1 = 4\}$, $\{a_2 = 0, m_2= 1\}$, and $\{a_3 = 90, m_3 = 1 \}$. Let the vehicles that travel from A to C, and from B to C be all regular vehicles, i.e. $\alpha_{\text{AC}} = \alpha_{\text{BC}} =0$. Figure~\ref{fig:triangular_net} shows a plot of the network social delay versus the fraction of autonomous vehicles traveling along O/D pair AB, $\alpha_{\text{AB}}$. As the figure shows, as vehicle autonomy increases, so does the social delay. Note that the social delay is unique in this case. This example shows that existence of vehicle autonomy along certain network O/D pairs can result in worsening the overall or social delay of the network even if the road degrees of capacity asymmetry are homogeneous. This is of paramount importance in practice. For instance, if O/D pair AB belongs to a high--income neighborhood, autonomous vehicles may first be deployed along this path, while other neighborhood or O/D pairs may still travel via regular vehicles. Then, although the early adoption of autonomous vehicles along O/D pair AB will lead to a decrease in travel delay of O/D pair AB, it worsens the social delay in the network and increases the delays experienced by users along other O/D pairs. This example shows that even with homogeneous degrees of capacity asymmetry, when there exist multiple O/D pairs, different autonomy fractions along network O/D pairs can be another source of heterogeneity in the network; hence, counterintuitive behaviors might occur for networks with mixed autonomy. 
\end{example}

It was shown in~\cite{fisk1979more,dafermos1984some} that a decrease in the total demand of a single O/D pair, might lead to an increase in delay of travel along other network O/D pairs and the social delay. In the previous example, we showed that a similar behavior can also be observed due to the presence of autonomous vehicles. In fact, what we have shown so far is that the long known paradoxical traffic behavior resulting from constructing more roads or reducing demands can actually happen in networks with mixed autonomy due to the presence of autonomous vehicles. Thus, the mobility benefits of increasing autonomy in a network are not immediate, and
in order to take advantage of the full mobility potential of autonomous vehicles, control and routing strategies that guarantee mobility benefits must be developed for the next generation of traffic networks.

Now that we have shown, the social delay can increase as a consequence of the presence of autonomous vehicles in networks with multiple O/D pairs, we wish to study whether we can bound this degradation in the network performance, to see how much worse the social delay can  get with increasing the fraction of autonomous vehicles. To answer this, we derive a bound on the performance degradation that can result from all possible demand and autonomy fraction vectors in general networks that have a \emph{homogeneous} degree of capacity asymmetry. To this end, for a given network $G$ and a demand vector $r$, define the vector of fictitious reduced demand $\tilde{r} = (\tilde{r}_w :w \in W)$ to be $\tilde{r}_w = (1-\alpha_w) r_w + \mu \alpha_w r_w $ for each O/D pair $w \in W$. Consider an auxiliary fictitious routing game with a total demand $\tilde{r}$ of only regular vehicles on $G$. For this auxiliary game, similar to Theorem~\ref{theo:homog}, define $(\tilde{e}_l: l \in L)$ to be

\begin{align}\label{eg:aux-delay}
\tilde{e}_l =  a_l \left(1+\gamma_l\left(\frac{\tilde{f}_l}{m_l} \right)^{\beta_l}\right),
\end{align}
and let $\tilde{e}_w (\tilde{r})$ be the delay of travel for each $w \in W$ in this auxiliary game. Then, using the auxiliary fictitious game, we can state the following proposition.


\begin{figure}
\centering
 \begin{tikzpicture}[scale=0.8]
\begin{axis}
[
no markers,
axis x line=bottom,
axis y line=left,
xmax=1, ymax = 10830,
xmin=0, ymin=10670,
xlabel={$\alpha_{AB}$},
ylabel={$J$},
]


\addplot[MidnightBlue,line width=2pt] coordinates {
(0., 10676.) (0.1, 10691.3) (0.2, 10706.6) (0.3, 10721.9) (0.4, 10737.2) (0.5, 10752.5) (0.6, 10767.8) (0.7, 10783.1) (0.8, 10798.4) (0.9, 10813.7) (1., 10829.)
};
\end{axis}
\end{tikzpicture}
\caption{Social delay in Example 5 for different fraction of autonomous vehicles traveling along O/D pair $AB$ when vehicles along all other O/D pairs are regular.}
\label{fig:triangular_net}
\end{figure}
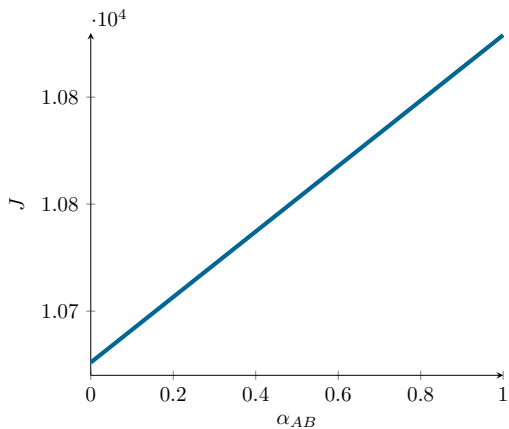

\begin{proposition}\label{prop:social_delay_aux}
Consider a general network $G = (N,L,W)$ with a homogeneous degree of capacity asymmetry $\mu \leq 1$ in all of its links. For any demand vector $r$, fix the vector of autonomy fraction $\alpha = (\alpha_w:w \in W)$ such that $0 \leq \alpha_w \leq 1$ for all $w \in W$. Then, we have
\begin{enumerate}
\item The social delay $J(f)$ is unique for all Wardrop equilibrium flow vectors $f$.
\item The social delay of the original game is $J(f) = \sum_{w \in W} r_w \tilde{e}_w(\tilde{r}_w)$  for all Wardrop equilibrium flow vectors $f$. 
\end{enumerate}
\end{proposition}


\begin{proof}
Fix $r$ and $\alpha$, such that for each $w \in W$, $0 < r_w$ and $0 \leq \alpha_w \leq 1$. Recalling Lemma~\ref{remark:exist}, we know that there exists at least one equilibrium. Let $f = (f_p^r, f_p^a:p \in \PP)$ be such an equilibrium flow vector for $G$. For each path $p \in \PP$, define $\tilde{f}_p := f_p^r + \mu f_p^a$. 
By generalizing the proof of Theorem~\ref{theo:homog}, it is easy to see that $\tilde{f} = (\tilde{f}_p :p \in \PP)$
 is an equilibrium for the defined auxiliary routing game on $G$ with reduced demand $\tilde{r}$ of only regular vehicles.
Moreover, for each path $p \in \PP$, $e_p(f) = \tilde{e}_p(\tilde{f})$. Therefore, for each O/D pair $w \in W$, $\tilde{e}_w(\tilde{f}) = \min_{p \in \PP_w} \tilde{e}_p(\tilde{f}) = \min_{p \in \PP_w} {e}_p({f}) = e_w(f)$. Hence,

\begin{align}\label{eq:equ_soci_delay}
J(f) = \sum_{w\in W} r_w e_w(f) = \sum_{w\in W} r_w \tilde{e}_w(\tilde{f}).
\end{align}

Since $\tilde{f}$ contains only regular vehicles, recalling Remark~\ref{rem:sing_type} and Proposition~\ref{prop:weak_uniq},  for each $w \in W$, the delay of travel $\tilde{e}_w(\tilde{f})$ is unique for a given $\tilde{r}$; thus, 

\begin{align}\label{eq:social_delay_final_reduced}
J(f) = \sum_{w\in W} r_w \tilde{e}_w(\tilde{r}).
\end{align}

As $\tilde{r}$ is uniquely determined for a given demand vector $r$ and a vector of autonomy fraction $\alpha$, the social delay $J(f)$ is unique for all Wardrop equilibrium flow vectors $f$ and can be obtained via~\eqref{eq:social_delay_final_reduced}.
\end{proof}

The uniqueness of social delay established by Proposition~\ref{prop:social_delay_aux} implies that for a fixed demand vector $r$, the social delay is a well defined function of autonomy fraction $\alpha$. With a slight abuse of notation, we use $J(\alpha)$ to emphasize the dependence of the social delay on the vector of autonomy fraction $\alpha$. 
Note that Proposition~\ref{prop:social_delay_aux} establishes a connection between our original routing game, which has two classes of vehicles, with a fictitious auxiliary routing game, which has only regular vehicles and a reduced demand vector $\tilde{r}$. We exploit this connection in the remainder of the paper. Since the auxiliary game has only one class of vehicles, the results in~\cite{correa2008geometric} hold for this game. Before proceeding, we need to adopt and review some of the definitions in~\cite{correa2008geometric} for our proposed auxiliary game. 

In the auxiliary game, for a given O/D demand vector $\tilde{r}$, a flow vector $\tilde{f}$ is feasible if $\tilde{f}_p \geq 0$ for all paths $p \in \PP$, and $\sum_{p\in\PP_w}\tilde{f}_p = \tilde{r}_w$ for all $w \in W$. Let $\phi \in \mathbb{R}^{|L|}$ be a vector of \emph{link} flows that result from a feasible flow vector $\tilde{f}$, where $|L|$ is the number of links in the network. 
Also, let $\Phi$ represent the set of all feasible link flow vectors $\phi$ for a given reduced demand vector $\tilde{r}$. Then, for a vector of link delay functions $(\tilde{e}_l :l\in L)$ of the form~\eqref{eg:aux-delay} and any vector $v \in \Phi$, define

\begin{align}\label{eq:lambda_def}
\lambda \big( (\tilde{e}_l:l\in L), v \big) := \max_{x \in \mathbb{R}^{|L|}_{\geq 0}} \frac{\sum_{l \in L} \big(\tilde{e}_l(v_l) - \tilde{e}_l(x_l)\big) x_l}{\sum_{l \in L}\tilde{e}_l(v_l)v_l},
\end{align}
\noindent where $0/0$ is considered to be $0$. Additionally, let $\mathcal{E}$ be the class of delay functions represented by~\eqref{eg:aux-delay}. Define
\begin{align}\label{eq:lambda_def_final}
\lambda(\mathcal{E}) := \sup_{ (\tilde{e}_l:l\in L) \in \mathcal{E}, v \in \Phi} \lambda \left( (\tilde{e}_l:l\in L), v \right).
\end{align}
It is important to mention that since the class of delay functions $\mathcal{E}$
is monotone, $\lambda(\mathcal{E})\leq 1$ in our setting (See Section 4 in~\cite{correa2008geometric}). 
Note that $\lambda(\mathcal{E})$ can be easily computed for certain classes of delay functions such  as polynomials. For instance, $\lambda(\mathcal{E}) = \frac{1}{4}$ for the class of linear delay functions.

Now, we can bound the network performance degradation due to the introduction of autonomy in homogeneous networks via the following theorem.

\begin{theorem}\label{theo:bound}
Consider a general network $G = (N,L,W)$ with a homogeneous degree of capacity
asymmetry $\mu$. Fix the demand vector $r$. Let $J^o$ be the social delay when
all vehicles are nonautonomous, i.e. $\alpha_w = 0$ for all $w \in W$. Then, for any other vector of autonomy fraction $\alpha$ such that $0 \leq \alpha_w \leq 1$ for all $w \in W$, we have
\begin{align}
J(\alpha) \leq (1-\lambda(\mathcal{E}))^{-1} J^o,
\end{align}
where $J(\alpha)$ is the social delay for the vector of autonomy fraction $\alpha$. Note that using Proposition~\ref{prop:social_delay_aux}, $J(\alpha)$ and $J^o$ are unique, and; thus, well defined.   
\end{theorem}

\begin{proof}
Fix the demand vector $r$. Let $f^o = (f^o_p: p \in \PP)$ be an equilibrium flow vector when all vehicles are regular. 
We further use $f_l^o$ to denote the flow along link $l\in L$ in this
case. Note that using Proposition~\ref{prop:weak_uniq} , we know that $f_l^o$ is unique for every link $l \in L$. Moreover, for each path $p \in P$, we use $e_p^o$ to represent the delay
along path $p$ when all vehicles are regular. Using Remark~\ref{rem:sing_type}
and Proposition~\ref{prop:weak_uniq}, in the absence of autonomy, the delay of
travel for each O/D pair $w \in W$ is unique. Thus, in this case, the unique
social delay $J^o = \sum_{w\in W}r_w e^o_w(r)$ , where $e_w^o(r)$ is the delay of travel along $w \in W$ when all vehicles are regular.

On the other hand, when there are autonomous vehicles with a given autonomy
fraction $\alpha$ in the network, as defined in
Proposition~\ref{prop:social_delay_aux}, construct the auxiliary game on $G$
with fictitious reduced demand $\tilde{r}=(\tilde{r}_w : w \in W)$ of only
regular vehicles, where $\tilde{r}_w = (1-\alpha_w)r_w + \mu r_w \alpha_w$ for every $w \in W$. Let
$\tilde{f}=(\tilde{f}_p:p\in \PP)$ be an equilibrium flow vector for this
auxiliary game. Using Proposition~\ref{prop:social_delay_aux}, the social delay
of the network with autonomous vehicles is given by $J(\alpha) = \sum_{w \in W}
r_w \tilde{e}_w(\tilde{r})$. First, we claim that 
\begin{align}\label{eq:first_ineq_req}
J(\alpha) = \sum_{w \in W} r_w \tilde{e}_w({\tilde{r}}) 
\leq \sum_{l \in L} f_l^o \tilde{e}_l(\tilde{r}). 
\end{align}
To see this, note that for every link $l \in L$, we have $f_l^o = \sum_{p \in \PP: l\in p} f^o_p$. Furthermore, the origin and destination of each path $p \in \PP$ are unique. Hence, each path $p$ belongs to one and exactly one O/D pair $w \in W$. Consequently, $f_l^o = \sum_{w \in W}  \sum_{p\in\PP_w:l\in p} f^o_p$, and we have 
\begin{align*}
\sum_{l\in L} f_l^o \tilde{e}_l(\tilde{r}) &= \sum_{l\in L} \left( \sum_{w \in W}  \sum_{p\in\PP_w:l\in p} f^o_p \right) \tilde{e}_l(\tilde{r}) \\
&= \sum_{w \in W} \sum_{l\in L}\left( \sum_{p\in\PP_w:l\in p} f^o_p \right) \tilde{e}_l(\tilde{r}) \\
&= \sum_{w \in W} \sum_{p\in\PP_w} f_p^o \sum_{l: l\in p}\tilde{e}_l(\tilde{r}) \\
&= \sum_{w \in W} \sum_{p\in\PP_w} f_p^o \tilde{e}_p(\tilde{r}), 
\end{align*}
where $\tilde{e}_p(\tilde{r})$ is the delay of path $p \in \PP_w$ for the auxiliary game. Recalling Definition~\eqref{eq:delay_travel}, for the auxiliary game, the travel delay of an O/D pair $w \in W$ is given by $\tilde{e}_w(\tilde{r}) = \min_{p \in \PP_w} \tilde{e}_p(\tilde{r})$; thus, we have
\begin{align*}
 \sum_{w \in W} \sum_{p\in\PP_w} f_p^o \tilde{e}_p(\tilde{r}) &\geq   \sum_{w \in W} \sum_{p\in\PP_w} f_p^o \tilde{e}_w(\tilde{r})  \\
&=  \sum_{w \in W} \tilde{e}_w(\tilde{r})\sum_{p\in\PP_w} f^o_p \\
&= \sum_{w \in W} r_w \tilde{e}_w(\tilde{r}),
\end{align*}
\noindent which proves our claim in \eqref{eq:first_ineq_req}. Now,
since the auxiliary game has only one class of vehicles, we can use Lemma 4.1
from~\cite{correa2008geometric}. More precisely, since
$\tilde{f}$ is an equilibrium for the auxiliary game, then Lemma 4.1
from~\cite{correa2008geometric} states that for every nonnegative vector of link
flows $x \in \mathbb{R}^{|L|}_{\geq 0}$ ($x$ is not necessarily a feasible link flow vector), we have
\begin{align}\label{eq:lemma_orig}
\sum_{l\in L} x_l \tilde{e}_l(\tilde{f}_l)  \leq \sum_{l\in L}x_l \tilde{e}_l(x_l)  + \lambda(\mathcal{E}) \sum_{l\in L} \tilde{f}_l \tilde{e}_l(\tilde{f}_l). 
\end{align}
Since $f_l^o$ is nonnegative for every link $l \in L$, substituting $x_l$ by $f_l^o$ in~\eqref{eq:lemma_orig}, we get 
\begin{align} \label{eq:lemma}
\sum_{l\in L} f_l^o \tilde{e}_l(\tilde{f}_l) \leq \sum_{l\in L} f^o_l \tilde{e}_l(f^o_l) + \lambda(\mathcal{E}) \sum_{l\in L} \tilde{f}_l \tilde{e}_l(\tilde{r}).
\end{align}
Now, note that since both the auxiliary game and the game with no autonomy have only regular vehicles, utilizing~\eqref{eg:aux-delay}, we realize that
\begin{align*}
\tilde{e}_l (f_l^o) &= a_l \left(1 + \gamma_l \left(\frac{f_l^o}{m_l}\right)^{\beta_l}\right) \\
&= e_l^o (f_l^o).
\end{align*}
\noindent Thus, 
\begin{align}\label{eq:1}
\sum_{l\in L} f^o_l \tilde{e}_l(f^o_l)=\sum_{l\in L} f^o_l e^o_l(f^o_l) = J^o.
\end{align}
Now, since $J(\alpha) = \sum_{w \in W} r_w \tilde{e}_w({\tilde{r}})$ and for all links $l \in L$, $\tilde{e}_l(\tilde{r}) = \tilde{e}_l(\tilde{f}_l)$ by definition, using~\eqref{eq:first_ineq_req},~\eqref{eq:lemma}, and~\eqref{eq:1}, we realize that
\begin{align}\label{eq:reducd_lemma}
J(\alpha) \leq J^o + \lambda(\mathcal{E}) \sum_{l\in L} \tilde{f}_l \tilde{e}_l(\tilde{r}).
\end{align}
As $\tilde{f}$ is an equilibrium for the auxiliary routing game,
$\sum_{l \in L}\tilde{f}_l \tilde{e}_l (\tilde{r})= \sum_{w \in W} \tilde{r}_w
\tilde{e}_w (\tilde{r})$. Since for each O/D pair $w \in W$, $\alpha_w \leq 1$,
we have $\tilde{r}_w \leq r_w$. Therefore, using Proposition~\ref{prop:social_delay_aux},
\begin{align}\label{eq:auxi_eq}
\sum_{w\in W} \tilde{r}_w \tilde{e}_w (\tilde{r}) \leq \sum_{w\in W} r_w \tilde{e}_w (\tilde{r})=J(\alpha).
\end{align}
Using~\eqref{eq:auxi_eq} and~\eqref{eq:reducd_lemma}, we get
\begin{align}
J(\alpha) \leq J^o + \lambda(\mathcal{E}) J(\alpha).
\end{align}
Hence, for the our monotone class of delay functions $\mathcal{E}$ with $\lambda(\mathcal{E})< 1$, we can conclude that
\begin{align*}
J(\alpha) \leq (1-\lambda(\mathcal{E}))^{-1} J^o,
\end{align*}
which completes the proof.
\end{proof}

Theorem~\ref{theo:bound} provides an upper bound on the severity of increases in traffic delays when a fraction of regular vehicles is replaced by autonomous vehicles. 

We now postulate, as an analogous concept to the price of anarchy~\cite{roughgarden2002bad}, the price of vehicle autonomy in homogeneous networks under every demand vector $r$ as follows:


\begin{align}\label{eq:autno_price}
\eta := \max_{\alpha:\; 0 \leq \alpha_w \leq 1, \, \forall w} \frac{J(\alpha)}{J^o},
\end{align}
 
 Theorem~\ref{theo:bound} states that $\eta \leq (1-\lambda(\mathcal{E}))^{-1}$.
For polynomial delay functions of degree less than or equal to 4, $(1-\lambda(\mathcal{E}))^{-1} = 2.151$~\cite{correa2008geometric}. 
 Interestingly, the bound that we have derived for the price of vehicle autonomy is similar to the bounds derived for the price of anarchy of routing games with a single class of users in~\cite{roughgarden2002bad,correa2008geometric}. Note that this bound for $\eta$ is different from the price of anarchy of routing games with mixed autonomy~\cite{lazar2017price}, it is similar to that of routing games with only a single class of vehicles. However, unlike the bounds for price of anarchy, the tightness of our bound for $\eta$ must be further investigated.




\section{Conclusion and Future Work} \label{sec:future}

In this paper, we studied how the coexistence of autonomous and regular vehicles in traffic networks will affect network mobility when all vehicles select their routes selfishly. We compared the total social network delay at a Wardrop equilibrium in networks with mixed autonomy with that of the networks with only regular vehicles. Having shown that the equilibrium is not unique in the mixed--autonomy setting, we proved that the total social delay is unique when the road degree of capacity asymmetry, which is the ratio between the roadway capacity with only regular vehicles and the roadway capacity with only autonomous vehicles, is homogeneous among its roadway. We further proved that the total social delay is a nonincreasing and continuous function of the fraction of autonomous vehicles on the roadways (aka the autonomy ratio $\alpha$) when the network has only one O/D pair. However, we showed that allowing for heterogeneous degrees of capacity asymmetry or multiple O/D pairs in the network results in counter intuitive behaviors such as the fact that increasing network autonomy ratio can worsen the network total social delay. Finally, we derived an upper bound for the ``price of vehicles autonomy" in networks with a homogeneous degree of capacity asymmetry, which estimates the worst possible increase in network social delay, due to the introduction of autonomous vehicles.

We believe that the results presented in this paper indicate that the expected mobility benefits resulting from wide spread utilization of autonomous vehicles in traffic networks are not immediate. Thus, in order to take advantage of the potential mobility benefits of autonomy, it will be necessary to study the stability of traffic equilibria in networks with mixed autonomy. Once the stable system equilibria are characterized, traffic management and control strategies must be developed for the traffic network that are guaranteed to steer the system to the equilibria that have lower total delay. Therefore, revisiting routing and tolling strategies for networks with mixed vehicle autonomy is essential.
\section*{Acknowledgment}

This work is supported by the National Science Foundation under Grant CPS 1545116.

\ifarxiv
\relax
\fi



%

\newcommand{\etalchar}[1]{$^{#1}$}

\ifarxiv
\bibliographystyle{alpha}
\fi



%








\end{document}